\documentclass{amsart}
\usepackage{fullpage}
\usepackage{amsmath,amsfonts,amssymb,amsthm} 
\newtheorem{thm}{Theorem}[section]
\newtheorem{dfn}[thm]{Definition}
\newtheorem{lem}[thm]{Lemma}
\newtheorem{prp}[thm]{Proposition}

\theoremstyle{definition}\newtheorem{rmk}[thm]{Remark}
\theoremstyle{definition}\newtheorem{example}[thm]{Example}

\newcommand\mc{\mathcal}
\newcommand\ol{\overline}
\newcommand\Dslash{{D\mkern-11.5mu/\,}}

\DeclareMathOperator{\tr}{Tr}
\DeclareMathOperator{\Ad}{Ad}
\DeclareMathOperator{\ad}{ad}
\newcommand\cinf[1]{C^{\infty}(#1)}
\newcommand\spectraltriple{(\Gamma(M,B), L^2(M,B\otimes S),D_B)}

\newcommand\spectraltripleJeven{(\Gamma(M,B), L^2(M,B\otimes S),D_B,J,\gamma_B)}
\renewcommand\inf\infty

\renewcommand\epsilon\varepsilon

\def\bA{\mathbb{A}}
\def\A{\mathcal{A}}
\def\C{\mathbb{C}}
\def\ch{\textup{ch}}
\def\H{\mathcal{H}}
\DeclareMathOperator{\ind}{Index}
\def\top{\textup{top}}

\begin{document}
\title{The noncommutative geometry of Yang--Mills fields}
\author{Jord Boeijink and Walter D. van Suijlekom}
\address{Institute for Mathematics, Astrophysics and Particle Physics, Faculty of Science, Radboud University Nijmegen, Heyendaalseweg 135, 6525AJ Nijmegen, The Netherlands}
\email{j.boeijink@math.ru.nl; waltervs@math.ru.nl}

\begin{abstract}
We generalize to topologically non-trivial gauge configurations the description of the Einstein--Yang--Mills system in terms of a noncommutative manifold, as was done previously by Chamseddine and Connes. Starting with an algebra bundle and a connection thereon, we obtain a spectral triple, a construction that can be related to the internal Kasparov product in unbounded KK-theory. In the case that the algebra bundle is an endomorphism bundle, we construct a $PSU(N)$-principal bundle for which it is an associated bundle. The so-called internal fluctuations of the spectral triple are parametrized by connections on this principal bundle and the spectral action gives the Yang--Mills action for these gauge fields, minimally coupled to gravity. Finally, we formulate a definition for a topological spectral action.
\end{abstract}

\date{\today}
\maketitle
\section{Introduction}
One of the main applications of noncommutative geometry to theoretical physics is in deriving the Yang--Mills action from purely geometrical data \cite{ChamseddineConnes}. In fact, the full Lagrangian of the Standard Model of high-energy physics -- including the Higgs potential -- can be derived by starting with a noncommutative Riemannian spin manifold \cite{CCM07}. 

It is interesting to confront this with the geometrical approach to Yang--Mills theory (\textit{cf}. \cite{Ati79}), using the language of principal fiber bundles and connections thereon. It turns out that the noncommutative geometrical description of \cite{CC97} corresponds to topologically trivial $SU(N)$-principal bundles. It is the goal of this paper to generalize this to topologically non-trivial gauge configurations. As a matter of fact, we derive the Yang--Mills action for gauge fields defined on a non-trivial principal bundle from a noncommutative Riemannian spin manifold, that is, from a spectral triple. 
Since spectral triples -- and more generally, (unbounded) KK-theory -- form a natural setting for doing index theory, our construction has potential applications to {\it e.g.} the study of moduli spaces of instantons in noncommutative geometry. 

Our construction will naturally involve algebra bundles and connections thereon, for which -- after some preliminaries -- we will give their definition in Section \ref{sect:algebrabundles}. There, we will also construct a spectral triple from this data. The above connection plays the same role as it does in the internal Kasparov product in KK-theory and we will explore this relation in some detail in Section \ref{sect:KK}.

In the case that the algebra bundle has typical fiber $M_N(\C)$ -- {\it i.e.} it is an endomorphism bundle -- it is possible to construct a $PSU(N)$-principal bundle, with the algebra bundle as an associated bundle. We will explore this case in Section \ref{sect:ym}. The so-called internal fluctuations of the above spectral triple are parametrized by connections on this principal bundle. Finally, we show that the spectral action principle applied to the spectral triple gives the Yang--Mills action on a topologically non-trivial $PSU(N)$-principal bundle, minimally coupled to gravity.

In the concluding section, we sketch the definition of a so-called topological spectral action.

\subsection*{Acknowledgements}
We thank Simon Brain for a careful proofreading of the manuscript, as well as valuable suggestions and remarks.

\section{Preliminaries}
\subsection{Spectral triples and the spectral action principle}
Spectral triples, as they are introduced in \cite{Connes} are at the heart of noncommutative geometry. In fact, they generalize $spin^c$-structures to the noncommutative world. 
%They are to Riemannian geometric spaces what $C^*$-algebras are to topological spaces. Indeed, through Gelfand duality we can equally well describe a locally compact Hausdorff topological spaces by a commutative $C^*$-algebra. If a topological space $M$ has a smooth manifold structure, one cannot recover more than the topology from the the $C^*$-algebra $C_0(M)$. In \cite{???} Connes suggested that a same sort of result as the Gelfand duality was true for a smooth Riemannian $spin^c$-manifold when one considers in addition to $C_0(M)$ also the Dirac operator, consistuting the so-called \emph{canonical spectral triple} ({\it cf.} Example \ref{ex:canonicaltriple}). In 2008, Connes showed
%in \cite{ConnesReconstruction} that, conversely, to any spectral triple with commutative algebra $\mc{A}$ (plus some extra technical conditions) there exists a
%$spin^c$-manifold $M$ such that $\mc{A} \cong \cinf{M}$. In the same way as a noncommutative $C^*$-algebra is considered as a noncommutative topological space, a spectral triple can be
%considered as a noncommutative geometric space.
\begin{dfn}[\cite{Connes}]
A \emph{spectral triple} $(\mathcal{A}, \mathcal{H}, D$) is given by an involutive algebra $\mathcal{A}$ represented faithfully on the Hilbert space
$\mathcal{H}$, together with a densely defined, self-adjoint operator $D$ on $\mathcal{H}$ with the following properties:
\begin{itemize}
\item The resolvent operators $(D - \lambda)^{-1}$ are compact on $\mathcal{H}$ for all $\lambda \notin \mathbb{R}$,
\item For all $a \in \mathcal{A}$ the operator $[D,a]$ extends to a bounded operator defined on $\mathcal{H}$ .
\label{dfn:spectraltriple}
\end{itemize}
The triple is said to be \emph{even} if there exists an operator $\Gamma$ on
$\mathcal{H}$ with the properties
\begin{equation*}
\Gamma^*=\Gamma, \quad \Gamma^2=1, \quad \Gamma D + D \Gamma = 0, \quad \Gamma a - a \Gamma = 0.\label{eq:eventriple}
\end{equation*}
If such an operator does not exist, then the triple is said to be \emph{odd}.
\end{dfn}

\begin{example}
\label{ex:canonicaltriple}
The motivating example for the definition of a spectral triple is formed by the \emph{canonical triple}
\begin{equation*}
(\cinf{M}, L^2(M,S), \Dslash)
\end{equation*}
associated to any compact Riemannian spin-manifold $M$.\footnote{Here and in what follows we work in the category of smooth manifolds.}  The Hilbert space $L^2(M,S)$ consists of square-integrable sections of the spinor bundle $S \to M$.
The operator $\Dslash$ is the Dirac operator on the spinor bundle. For even dimensional spin-manifolds there exists a grading $\gamma$ on $L^2(M,S)$.
% that is given on $\Gamma(M,S)$ by  $(\gamma \cdot s)(x) = \hat{\gamma}_x (s(x))$, where $\hat{\gamma} \in \mathbb{C}l(M,g)$ is locally defined by $\hat{\gamma(x)} = i^{?} e_1(x) \cdot \cdots \cdot e_n(x)$, where $\{e_i\}$ is a local oriented orthonormal frame of $TM$.\textbf{maak referentie!}
\end{example}

A spectral triple can have additional structure such as reality.
\begin{dfn}[\cite{ConnesMarcolli}, Definition 1.124]
A \emph{real structure} on a spectral triple $(\mc{A}, \mc{H},D)$ is an anti-unitary operator $J: \mc{H} \rightarrow \mc{H}$, with the property that
\begin{equation*}
J^2 = \epsilon, \quad JD = \epsilon' DJ, \quad and \quad J\gamma = \epsilon'' \gamma J, \text{ (even} \text{ case}),
\end{equation*}
where the numbers $\epsilon$,$\epsilon'$, $\epsilon''$ are $\pm 1$. Moreover, there are the following relations between $J$ and elements of $\mc{A}$:
\begin{equation}
[a,b^0] = 0, \qquad [[D,a], b^0] = 0 \text{ for all } a,b \in \mathcal{A}.
\label{eq:jreq1}
\end{equation}
where $b^0 = J b^* J^{-1} \text{ for all } b \in \mathcal{A}$. A spectral triple $(\mathcal{A},\mathcal{H},D)$ endowed with a real structure $J$ is called a \emph{real spectral
triple}.
\label{dfn:realstructure}
\end{dfn}
The signs $\epsilon, \epsilon'$ and $\epsilon''$ determine the so-called KO-dimension (modulo 8) of the real spectral triple (see \cite{ConnesGravity} for more details).

\begin{example}
\label{ex:canonical}
For a spin-manifold and a given spinor bundle $S$ there exists an operator $J_M$ -- called charge conjugation -- on $L^2(M,S)$ such that
\begin{equation*}
(\cinf{M}, L^2(M,S), \Dslash, J_M)
%\label{eq:canonicaltripleJ}
\end{equation*}
is a real spectral triple. Here the $KO$-dimension is equal to the dimension of the spin-manifold $M$. For more details on the construction of $J_M$ the reader is referred to {\it e.g.} \cite{Varilly}. When the dimension $n$ is even, the inclusion of the grading operator $\gamma$ of Example \ref{ex:canonicaltriple} to the datum
\begin{equation}
(\cinf{M}, L^2(M,S), \Dslash, J_M, \gamma)
\label{eq:canonicaltripleJeven}
\end{equation}
yields a real and even spectral triple.
\end{example}

\begin{rmk}
Note that the existence of a real structure $J$ turns $\mc{H}$ into a bimodule over $\mc{A}$. Indeed, condition (\ref{eq:jreq1}) implies that the right action of $\mc{A}$ on $\mc{H}$ defined by
\begin{equation*}
\xi a := Ja^*J^* \xi, \quad (\xi \in \mc{H}, a \in \mc{A})
\end{equation*}
commutes with the left action of $\mc{A}$.
\end{rmk}

\subsubsection{Spectral triples and gauge theories}
In this subsection we show how noncommutative spectral triples naturally give rise to gauge theories, following \cite{ConnesGravity}. First of all, note that the most natural notion of equivalence of (unital) noncommutative
($C^*$-)algebras is Morita equivalence (\cite{Rieffel}). A unital algebra $\mc{A}$  is {\it  Morita equivalent} to a unital algebra $\mc{B}$ if and
only if there exists a $\mc{B}-\mc{A}$-module $\mc{E}$ which is finitely generated and projective as an $\mc{A}$-module such that $\mc{B} = \text{End}_{\mc{A}} \mc{E}$. Commutative algebras are Morita equivalent if and only if they are isomorphic, justifying this notion of equivalence for noncommutative algebras.

If $(\mc{A}, \mc{H}, D,J,\gamma)$ is a real and even spectral triple and $\mc{B}$ is a unital algebra Morita equivalent to $\mc{A}$, then there is natural way to construct a real and even spectral triple for the algebra $\mc{B}$. If this is done for the case $\mc{B} = \mc{A}$ (any algebra is Morita equivalent to itself) through the module $\mc{E} = \mc{A}$, the obtained spectral
triple is of the form
\begin{equation*}
(\mc{A}, \mc{H}, D_A := D + A + \epsilon' JAJ^{-1}),
\end{equation*}
where $A = \sum_j a_j [D,b_j]$ for $a_j, b_j \in \mc{A} $ is a bounded self-adjoint operator on $\mc{H}$ (see \cite{ConnesMarcolli}, Section 10.8 for more details). It is a
straightforward verification that this is again a real and even spectral triple. Thus we get another spectral triple consisting of the same algebra and Hilbert space but with
the operator $D$ fluctuated by an element $A$. Note that the terms $A$ and $JAJ^{-1}$ cancel each other if $\mc{A}$ is a commutative algebra (this follows by a small
calculation using the compatibility conditions of $J$ with $D$ and the action of $\mc{A}$). The occurrence of fluctuations of the operator $D$ by Morita equivalences is
therefore a purely noncommutative phenomenon. The element $A$ will be interpreted as the \emph{gauge potential}.

The \emph{gauge group} of the triple $(\mc{A}, \mc{H}, D)$ is the subgroup $\text{Inn}(\mc{A})$ of $*$-automorphism of $\mc{A}$ consisting of all automorphisms of the form $a \mapsto uau^*$ where $u \in \mc{A}$ satisfies $uu^*=u^*u=1$ (\cite{ConnesMarcolli}, Section 9.9). This inner automorphism group acts naturally on the constituents of a spectral triple as an intertwiner. The gauge potential transforms accordingly as $A \mapsto uAu^* + u[D,u^*]$. The action of the gauge group on the Hilbert space is given by  $\psi \mapsto uJuJ^* \psi$, where $\psi \in \mc{H}$ and $u \in \mc{U}(A)$.

\subsubsection{Spectral action principle}
\label{ssct:spectralaction}
Associated to a spectral triple we have a gauge group, a gauge potential and gauge transformations and in this way a spectral triple forms the setting of a gauge theory. To obtain the dynamics of the theory, the spectral action principle \cite{ConnesGravity,ChamseddineConnes} is used to calculate an action from the spectral triple. The action consists of two parts: the first part is a fermion part, which is defined by
\begin{equation*}
S_f[\psi, A] = \langle \psi, D_A \psi \rangle,
\end{equation*}
where $\langle \cdot, \cdot \rangle$ denotes the inner product on $\mc{H}$. Note that the fermionic action depends on the gauge potential $A$ in $D_A$ but that it is invariant under gauge transformations. The other part of the action is the bosonic action which is defined by
\begin{equation}
\label{eq:spectralaction}
S_b[A] =\tr(f(D_A / \Lambda)),
\end{equation}
where $\tr$ denotes the trace in $\H$, $f$ is a suitable cut-off function with $\Lambda > 0$. Note that, just as for the fermionic action, the expression of $S_b$ is invariant under the transformations $D_A \mapsto uJuJ^* D_A Ju^*Ju^*$ for $u \in \mc{A}$ unitary. 

\subsection{Einstein--Yang--Mills theories and spectral triples}
\label{ssct:EYM}
Chamseddine and Connes showed in \cite{ChamseddineConnes} that
Yang--Mills gauge theory over a compact Riemannian spin-manifold $M$ can be obtained from a spectral triple built from the canonical triple associated to this manifold and a matrix algebra. In this subsection we will briefly review their results and we will relate it to the description of gauge theories in terms of principal fiber bundles. Let us first recall how such fiber bundles enter gauge theories.

%Therefore, before we briefly summarize the results of Chamseddine and Connes, we give the definition of a gauge theory on a manifold $M$. In this paper we only study gauge theories where the Lie group $G$ is a linear Lie group. In this case it suffices to take the following definition of a gauge potential.

\begin{dfn}
\label{dfn:connection1form}
Let $G$ be a matrix Lie group and let $P$ be a principal $G$-bundle. A \emph{connection} $\omega$ assigns to each local trivialization $\phi_U: \pi^{-1}(U) \rightarrow U \times G$ a $\mathfrak{g}$-valued one-form $\omega_u$ on $U$. If $\phi_V$ is another local trivialization and $g_{uv}: U\cap V \rightarrow G$ is the transition function from $(U,\phi_U)$ to $(V,\phi_V)$, then we require the following transformation rule for $\omega$:
\begin{equation}
\omega_u = g_{uv}^{-1} dg_{uv} + g^{-1}_{uv} \omega_v g_{uv}. \label{eq:76}
\end{equation}
\end{dfn}
More generally, that is in the case of arbitrary Lie groups, the gauge potential is defined as a global $\mathfrak{g}$-valued connection 1-form on $P$ satisfying some extra
conditions. In the case of matrix Lie groups this definition coincides with Definition \ref{dfn:connection1form} (see for instance \cite{Bleecker} for more details). The
local one-forms $\omega_u$ are the gauge potentials one encounters in physics. %One now has the following definition of a gauge theory.
\begin{dfn}
\label{dfn:gaugetheory}
A \emph{gauge theory} with group $G$ over a manifold $M$ consists of a principal $G$-bundle together with a connection 1-form $\omega$ on $P$. %We speak of a gauge theory on a trivial bundle if $P$ is topologically trivial. 
The connection 1-form $\omega$ on $P$ is also called the \emph{gauge potential}.
If $G=(P)SU(N)$ then the gauge theory is called a $(P)SU(N)$-Yang--Mills theory.
\end{dfn}

We now briefly summarize the results of \cite{ChamseddineConnes} that obtained Yang--Mills theory on a manifold $M$ from a well-chosen spectral triple. From now on the manifold $M$ is assumed to be a compact $4$-dimensional spin-manifold.

Consider the following objects:
\begin{gather}
\mathcal{A} = C^{\infty}(M) \otimes M_N(\mathbb{C}) %\cong C^{\inf}(M, M_N(\mathbb{C}))
, \qquad
\mathcal{H} = L^2(M,S) \otimes M_N(\mathbb{C}), \qquad
D = \Dslash_M \otimes 1, \nonumber \\
J= J_M  \otimes ( \cdot)^*, \qquad
\gamma = \gamma_5 \otimes 1.
\label{eq:conclusion0}
\end{gather}
The bundle $S$ is the spinor bundle whose fibers are isomorphic to $\mathbb{C}^4$ as in Example \ref{ex:canonicaltriple} and the operator $\Dslash$ is the Dirac operator on
the bundle $S$. %The operator $J_M$ is the real structure belonging to the spectral triple $(\cinf{M}, L^2(M,S), \Dslash)$ and $T: M_N(\mathbb{C}) \rightarrow M_N(\mathbb{C})$ sends a matrix to its adjoint. 
Observe that this triple forms a spectral triple, being the product of the canonical triple $(\cinf{M}, L^2(M,S), \Dslash)$ and the matrix algebra $M_N(\mathbb{C})$ that acts on itself by left multiplication. We will describe this product structure in more detail in Subsection \ref{ssct:kkproduct}. 

The spectral triple in Equation \eqref{eq:conclusion0} is real and even  using the fact that the canonical triple \eqref{eq:canonicaltripleJeven} is real and even. %We will now use the heat expansion to calculate the spectral action of the fluctuated operator $D_A$ where $D = \Dslash \otimes 1$.
Let us now determine the fluctuated Dirac operator $D_A = D + A + \epsilon'JAJ^*$ for this spectral triple. The fact that $\epsilon' =1$ in 4 dimensions implies that 
\begin{eqnarray}
A + \epsilon' JAJ^* = \gamma^{\mu} A_{\mu} + J \gamma^{\mu} A_{\mu} J^*.
\label{eq:bpp}
\end{eqnarray}
In even dimensions one has
\begin{equation*}
J_M \gamma^{\mu} J_M^* = -\gamma^{\mu},
\end{equation*}
and if we use that left-multiplication by $JA_{\mu}J^*$ is right multiplication by $A^*_{\mu}$, Equation ($\ref{eq:bpp}$) turns into
$A + JAJ^* = \gamma^{\mu} \cdot \text{ad}(A_{\mu})$, since $A$ is self-adjoint. Thus the fluctuated Dirac operator is of the form:
\begin{equation}
D_A = D + i \gamma^{\mu} \mathbb{A}_\mu
%\label{eq:diractrans}
\end{equation}
where $\bA_{\mu} = -i\ad A_{\mu}$. The self-adjointness of $A$ implies that $\bA_\mu$ is an anti-hermitian one-form. Since $A$ acts in the adjoint
representation the $u(1)$-part drops out and we effectively have a $su(N)$-gauge potential. %Thus, we have obtained a $su(N)$-valued gauge potential $-iA_{\mu}$ that acts in the \emph{adjoint representation} on the spinors. The bosonic part of the spectral action is now determined with respect to the operator $D + \mathbb{A}$. 

It was shown in \cite{ChamseddineConnes} that the spectral action applied to the above spectral triple \eqref{eq:conclusion0} describes the Einstein--Yang--Mills system. It contains the Einstein-Hilbert action and higher-order gravitational terms, as well as the Yang--Mills action for a global
$su(N)$-valued 1-form $A_{\mu}$. This is in line with the interpretation of the fluctuation $A$ as a gauge potential. Comparing this with the definition of a
$PSU(N)$-Yang--Mills theory as in Definition \ref{dfn:gaugetheory}, the fact that the gauge potential $A_{\mu}$ is globally an $su(N)$-valued 1-form means that this corresponds to a gauge theory with a trivial principal $PSU(N)$-bundle $P$. 
The goal of this paper is to generalize the spectral triple \eqref{eq:conclusion0} in such a way that it determines a topologically nontrivial $PSU(N)$-gauge theory. 

\section{Algebra bundles and spectral triples}
\label{sect:algebrabundles}
%In section \label{spectralaction} we have calculated the spectral action for the real and even spectral triple $(\cinf{M} \otimes M_N(\mathbb{C}), L^2(M,S) \otimes
%M_N(\mathbb{C}), \Dslash \otimes 1, J_M \otimes T, \gamma^5 \otimes 1)$. The spectral action included a term containing a Yang--Mills Lagrangian with respect to some $\mathfrak{su}(N)$-valued gauge field $A$. 

In this section we will generalize the above spectral triple to obtain a gauge theory on a non-trivial $PSU(N)$-bundle. The
important observation here is that in the trivial case we started with the algebra $\cinf{M} \otimes M_N(\mathbb{C})$ which is precisely the algebra of sections of a trivial
$M_N(\mathbb{C})$-bundle over $M$. This suggests for the non-trivial case that the algebra in the spectral triple is given by $\Gamma(M,B)$, where $B$ is an arbitrary locally
trivial algebra bundle where the fiber is the $*$-algebra $M_N(\mathbb{C})$. In fact, we will construct such a real and even spectral triple $(\mc{A}, \mc{H}, D, J, \gamma)$ where the algebra $\mc{A}$ is isomorphic to $\Gamma(M,B)$. This allows for a derivation of Yang--Mills theory for a gauge connection on a non-trivial principal fiber bundle in the next section. 

%from this spectral triple. This is done in the next section where we will establish the analogue of the Serre-Swan theorem for ($*$-)algebra bundles.
%\end{itemize}

%In order to avoid confusion concerning algebra bundles and to set notation, we will first introduce them in some details in Section \ref{ssct:definitionalgebrabundle}. %There are some subtleties in the definition of an algebra bundle concerning its local triviality and we deal with these subleties in such a way that the algebra bundle-analogue of the Serre-Swan theorem is true.  In Subsection we use this analogue of the Serre-Swan theorem to construct a spectral triple to every locally trivial $*$-algebra bundle with fiber $(M_N(\mathbb{C})$ that generalizes the spectral triple of equation (\ref{eq:conclusion0}). The form of this generalized spectral triple is put in a more general mathematical framework in \ref{ssct:kkproduct} using the $KK$-product of unbounded Kasparov modules. Finally, in Section \ref{sct:ymtnoncomm} we show using the spectral action principle how the spectral triple describes a non-trivial $PSU(N)$-gauge theory.

\subsection{Definition of algebra bundles}
\label{ssct:definitionalgebrabundle}
In this paper we take the following definition of an algebra bundle.
\begin{dfn}
An \emph{algebra bundle} $B$ is a vector bundle together with a vector bundle homomorphism $\mu: B \otimes B \rightarrow B$ such that for all $x \in M$:
\begin{equation}
%\label{eq:propertyalgebrabundle}
\mu(p_x \otimes(\mu(q_x \otimes r_x )) = \mu(\mu(p_x \otimes q_x) \otimes r_x), \quad \forall p_x,q_x, r_x \in B_x,
\end{equation}
inducing an associative algebra structure on each of the fibers of $B$ by setting $p_x  \cdot q_x = \mu(p_x
\otimes q_x)$ for two section $p$ and $q$ evaluated at $x \in M$. 

If $B_1$ and $B_2$ are two algebra bundles, then a map $\phi: B_1 \rightarrow B_2$ is called an \emph{algebra bundle morphism} if it is a vector bundle
morphism such that the restriction $f_{|(B_1)_x}: (B_1)_x \rightarrow (B_2)_x$ is an homomorphism of algebras.

An algebra bundle $B$ is called an {\rm involutive} or {\rm $\ast$-algebra bundle}, if there exists in addition an algebra bundle homomorphism $J: B \rightarrow \ol{B}^{op}$ such that $J^2=1$,\footnote{Here $\ol{B}^{op}$ is as a vector bundle conjugate to $B$ and has
opposite multiplication in the fibers} giving each fiber the structure of an involutive algebra by setting $p^*_x = J(p_x)$.

If $B_1, B_2$ are two $*$-algebra bundles, then an $*$-\emph{algebra bundle homomorphism} is a vector bundle homomorphism $f: B_1 \rightarrow B_2$ such that the restriction
$f_{|(B_1)_x}: (B_1)_x \rightarrow (B_2)_x$ is a $*$-algebra homomorphism for every base point $x \in M$. 

Let $AlgB(M)$ ($AlgB^*(M)$) denote the category whose objects are
all (involutive) algebra bundles (over $M$), and where the morphisms are all (involutive) algebra bundle morphisms.
\end{dfn}

\begin{rmk}
We note here that we do not require that the algebra in each fiber is the same. However, the way we introduced the associative algebra structures on the fibers guarantees
that the product of two smooth sections is again smooth.
%on the distinct fibers varies smoothly in the sense that the product of two smooth sections $s,t  \in \Gamma(M,B)$:
%\begin{equation}
%s\cdot t (x) = s(x)t(x) = \mu(s(x) \otimes t(x)) \quad (x \in M).
%\end{equation}
%is again a smooth section. 
This turns $\Gamma(M,B)$ into an associative algebra.
\end{rmk}

In general, the space of smooth sections of a vector bundle on $M$ is a module over $C^\infty(M)$. In the case of algebra bundles, this action is compatible with the multiplication in the fiber. Thus, if $B$ is an (involutive) algebra bundle, then $\Gamma(M,B)$ is a finitely generated (involutive) module algebra over $C^\infty(M)$. Recall that an $R$-module algebra is 
%The space of sections $\Gamma(M,B)$ of an algebra bundle or a $*$-algebra bundle carries the following algebraic structures.
%\begin{dfn}
%Let $R$ be a commutative ring. An \emph{$R$-module algebra} is 
an $R$-module $\mc{A}$ with an associative multiplication $\mc{A} \times \mc{A} \mapsto \mc{A}: (a,b) \mapsto
ab$ which is $R$-bilinear:
\begin{equation*}
r(ab)=(ra)b= a(rb) \quad \forall a,b \in \mc{A}, r \in R.
\end{equation*}
An $R$-module algebra is called involutive if there exists a map 
%We call such an $R$-module algebra finitely generated and projective if it is finitely generated and projective as an $R$-module. A homomorphism between two $R$-algebras is an $R$-linear map that preserves multiplication.
%\end{dfn}
%\begin{dfn}
%Let $R$ be a complex commutative $*$-algebra. An involution on an $R$-module algebra $\mc{A}$ is a map 
$*: \mc{A} \rightarrow \mc{A}$ such that %(here $*$ is used to denote the star operation in both $R$ and $\mc{A}$) the following equalities hold:
\begin{gather*}
(ab)^* = b^* a^* ; \qquad 
(a + b)^* = a^* + b^* ; \qquad
(r a)^* = r^* a^* ; \qquad \qquad (r,s \in R, a,b \in \mc{A}).
\end{gather*}
%If $R$ has a unit, the $R$-algebra $\mc{A}$ is also a complex algebra and $(\lambda a)^*=\ol{\lambda} a^*$ for $\lambda \in \mathbb{C}$ and $a \in \mc{A}$. A
%\emph{homomorphism} between two involutive $R$-algebras is an $R$-algebra homomorphism that preserves the star structure.
%\end{dfn}
%\begin{example}
%If $B$ is an algebra bundle, then $\Gamma(M,B)$ is a finitely generated projective $C^{\inf}(M)$-algebra and if $B$ is an $%*$-algebra bundle then $\Gamma(M,B)$ is a finitely
%generated projective $C^{\inf}(M)$-algebra with involution.
%\end{example}

Recall the well-known Serre--Swan Theorem for (complex) vector bundles over compact manifolds \cite{Swan}.
%To formulate the theorem, let $M$ be an arbitrary compact manifold.
%Suppose $f: E \mapsto F$ is a vector bundle homomorphism, then the map $\Gamma(f): s \mapsto f \circ s$ sends sections of the bundle $E$ to sections of the bundle $F$. The
%map $\Gamma: \text{Vect}_M \rightarrow \cinf{M}-\text{mod}$ given by
%\begin{eqnarray}
%E \mapsto \Gamma(M,E), \quad f \mapsto \Gamma(f),   \nonumber
%\end{eqnarray}
%is actually a covariant functor from the category of complex vector bundle over $M$ to the category of finitely generated projective $\cinf{M}$-modules.
\begin{thm}[Serre--Swan \cite{Swan}]
For every vector bundle $E$ over a compact manifold $M$ the space of sections $\Gamma(M,E)$ is a finitely generated projective $\cinf{M}$-module. The association $\Gamma$ of the space of sections to the vector bundle $E$ establishes an 
% and the functor $\Gamma$ establishes an
equivalence of categories between the complex vector bundles over $M$ and the category of finitely generated projective $\cinf{M}$-modules.
\end{thm}
%In fact, there is an explicit construction of a vector bundle $B$ out of some finitely generated projective $\cinf{M}$-module $\mc{P}$ such that $\Gamma(M,B) \cong \mc{P}$ as modules, see \cite{Swan}. %However, we will not be needing this explicit construction.
We now extend this result to arrive at an equivalence between ($*$)-algebra bundles and finitely generated projective $\cinf{M}$-module algebras (with involution). The idea is that the $\cinf{M}$-linear multiplicative structure on $\mc{P}$, where $\mc{P}$ is now considered as the space of sections of some vector bundle $B$ (which is unique up to
isomorphism), induces a product on the fibers $B_x$ such that $(s\cdot t)(x) = s(x) \cdot t(x)$. 
%\begin{lem}
%Let $E$ be a vector bundle over $M$. Then $\Gamma(M,E)$ is a free $\cinf{M}$-module if and only if $E$ is trivial.
%\label{lem:freetriv}
%\end{lem}
%\begin{proof}
%First, let $E = M \times V$ be a trivial vector bundle. Choose a basis $\{v_1, \dots ,v_k\}$ of $V$. Define the sections $e%_i(x) = (x,v_i)$, ($x \in M)$ for $i=1,\dots,k$.
%Then $\{e_1, \dots ,e_k\}$ is a free basis of the module $\Gamma(M,E)$. Hence, $\Gamma(M,E)$ is free.
%
%Conversely, let $\Gamma(M,E)$ be free. Define $\phi: E \rightarrow M \times V$ by setting $\phi\left(\sum_{i=1}^{k} \lambda%_i e_i(x)\right) = (x, \lambda_1, \dots
%,\lambda_m)$. This trivializes $E$.
%\end{proof}
The next lemma is crucial for lifting the multiplication structure on $\Gamma(M,B)$ to the fibers of $B$.
\begin{lem}\cite[Lemma 11.8b]{Nestruev}
\label{lem:sectionzero}
Let $\pi: B \rightarrow M$ be a vector bundle. Suppose $s$ is a section with $s(x)=0$ for some $x \in M$. Then there exist functions $f_i$ with $f_i(x)=0$ and sections $s_i
\in \Gamma(M,B)$ so that $s$ can be written as a finite sum $s = \sum_i f_i s_i$.
\end{lem}
%\begin{proof}%[Lemma \ref{lem:sectionzero}]
%Suppose that the bundle $E$ is trivial, then there exist a free basis $\{s_1, \dots , s_k\}$ of $\Gamma(M,B)$. Therefore, a% given section can be spanned over this basis as $s
%= \sum_{i=1}^k f_i s_i$ for certain $f_i \in \cinf{M}$. The equality $s(x)=0$ then implies that $f_i(x) = 0$ for all $i=1,\%dots ,k$, so in the trivial case the lemma is true.
%\footnote{the sections $s_i$ are never zero since $\{s_i\}$ forms a basis for $E_y$ at each point $y \in M$}.%
%For the nontrivial case, there exists a neighborhood $U \subset M$ of $x$, a finite number of sections $s_i \in \Gamma(U,B%)$ and functions $f_i \in \cinf{U}$ with
%$f_i(x)=0$, so that  $s_{|U} = \sum_i f_i s_i$ (this is the previous paragraph). Choose a function $f \in \cinf{M}$ such th%at $\supp f \subset U$ and $f(x)=1$. Now, extend
%the functions $ff_i \in \cinf{U}$ and the sections $fs_i \in \Gamma(U,B)$ to the whole of $M$ by setting them zero outside 5$U$. Now, we can write (keeping the notation $ff_i$
%and $fs_i$ for their extensions as well)
%\begin{equation}
%f^2 s = \sum_i (ff_i) (fs_i),
%\end{equation}
%so that
%\begin{equation}
%s = (1 - f^2)s + \sum_i (ff_i) (fs_i).
%\end{equation}
%Here we remark that $(1-f^2)(x), f_1(x), \dots ,f_k(x)$ are indeed zero.
%\end{proof}

Suppose that $\mc{P}$ is a finitely generated projective $\cinf{M}$-module which is at the same time an (involutive) $C^{\inf}(M)$-module algebra. The Serre--Swan Theorem gives a vector bundle $B$ so that $\mc{P} \simeq \Gamma(M,B)$ as $C^\infty(M)$-modules. We will now step-by-step introduce an (involutive) algebra bundle structure on $B$. 
%To prove this, we will use the fact that $\mc{P} \cong \Gamma(M,B)$ as $\cinf{M}$-modules. Therefore the module structure of $C^{\inf}(M)$ on $\mc{P}$ is given by
%\begin{equation}
%fs(x) = f(x)s(x), \quad f \in C^{\inf}(M), s \in \mc{P}.
%\end{equation}
%This module structure is needed to define the multiplication on the fibers $B_x$.
\begin{prp}
For $x \in M$, let $p,q \in B_x$ be given and suppose $s,t \in \mc{P}$ are such that $p = s(x)$ and $q = t(x)$. There exists a well-defined fiber multiplication $\mu(p \otimes q) : =st(x)$ turning $B$ into an algebra bundle. Consequently, we have $s t (y) = s(y) t(y)$ for all $y \in M$ and $s, t \in
\Gamma(M,B)$. 
\label{prp:dfnproduct}
\end{prp}
\begin{proof}
We need to show that the definition of the fiber product is independent of the choice of sections $s, t$ with $s(x) = p$ and $t(x)=q$. Therefore, let $s', t'$ be two other
sections of the bundle $B$ with $s'(x) = p$ and $t'(x) =q$. Then $s_0 = s' - s$ and $t_0 = t' - t$ are sections for which  $s_0(x)=t_0(x)=0$. According to Lemma
\ref{lem:sectionzero} $s_0$ and $t_0$ can be written as $s_0 = \sum_i f_i s_i$, $t_0 = \sum_i g_i t_i$ where $f_i(x)=g_i(x)=0$ for every $i$. This gives
\begin{eqnarray*}
s't' - st = (s' - s) t' + s(t' - t) = \sum f_i s_i t' + \sum_i g_i s t_i ,
\end{eqnarray*}
which evaluated at $x$ gives zero because of the module structure of $\Gamma(M,B)$. This argument shows that $s' t'(x) = s t(x)$ and the product is well-defined.  Actually,
the map $(s,t) \mapsto st$ is $\cinf{M}$-bilinear so it can be considered as a $\cinf{M}$-linear map from $\Gamma(M,B) \otimes_{\cinf{M}} \Gamma(M,B)$ to $\Gamma(M,B)$. This corresponds to a vector bundle homomorphism $\mu: B \otimes B \rightarrow B$.
%Because of the associativity of the product in $\Gamma(M,B)$, the product in the fibers will also be associative. 
If $\mc{P} \cong \Gamma(M,B)$ is unital with unit $1_P$,
then we can fix a unit in the fiber $B_x$ by setting $1_{B_x}=1_{\mc{P}}(x)$. 
\end{proof}

%If $\mc{P}$ is a finitely generated $C^{\inf}(M)$-algebra with involution, then in the same way as in the proof of the previous theorem one introduces an involution on the fibers of $B$.
\begin{prp}
For given $p \in B_x$, let $s \in \mc{P}$ be such that $s(x)=p$. Define $p^* := J p := s^*(x)$. This is a well-defined involutive structure on the fiber $B_x$, turning $B$ into an involutive algebra bundle.
\label{prp:dfnstar}
\end{prp}
\begin{proof}
We will use the same argument as before. Let $s'$ be another such section with $s'(x)=p$. Then with Lemma \ref{lem:sectionzero} $s_0 = s - s'$ can be written as a sum $\sum_i f_i s_i$ where $s_i \in \mc{P}$, $f_i\in C^{\infty}(M)$ and $f_i(x)=0$ for all $i$. This gives
\begin{equation*}
s^*(x) - s'^*(x) = (s - s')^*(x) = \sum_i (f_i s_i)^*(x) = \sum_i f^*_i(x) s^*_i(x) = 0,
\end{equation*}
so that the star structure is well-defined. That this is indeed a star structure on the fiber $B_x$ compatible with the algebra structure of the fiber, follows immediately
from the definition of a module $*$-algebra. %Example: the fact that $(pq)^* = q^*p^*$ is easy to proof. Choose $s,t \in \mc{P}$ such that $s(x)=p$ and $t(x) = q$. Then $(pq)^* = (st(x))^*  = (st)^*(x) = t^*s^*(x) = t^* (x) s^* (x) = q^* p^*$. The other properties follows similarly.
\end{proof}

The functor $\Gamma: \text{Vect}_M \rightarrow \text{FGP}_{\cinf{M}} \text{-mod}$ can be restricted to a functor $\hat{\Gamma}$ from the category $AlgB(M)$ of algebra bundles to the category of finitely generated projective $\cinf{M}$-algebras $\text{FGP}_{\cinf{M}} \text{-alg-mod}$. A similar statement applies to involutive algebra bundles and involutive module algebra. 
It follows from Propositions \ref{prp:dfnproduct} and \ref{prp:dfnstar} that the restricted functor $\hat{\Gamma}$ is still essentially surjective. As a restriction of a faithful functor, $\hat \Gamma$ is of course also faithful. To show that $\Gamma$ is full, let $B_1, B_2$ be two ($*$-)algebra
bundles and $F: \Gamma(M,B_1) \rightarrow \Gamma(M,B_2)$ be a ($*$-preserving) $\cinf{M}$-algebra-homomorphism. We will prove that %for given $F: \Gamma(M,B_1) \rightarrow \Gamma(M,B_2)$ 
the bundle homomorphism $\phi$ defined by
\begin{equation}
\phi(e) = F(s)(x), \quad (e  \in E_1),
\label{eq:full}
\end{equation}
where $s \in \Gamma(M,B_1)$ satisfies $s(x) =e$, is a $*$-algebra bundle homomorphism which is mapped to $F$ by $\hat{\Gamma}$. 

Firstly, observe that the map $\phi$ is
well-defined: let $s'$ be another section with $s'(x) = e$. Then $s - s' = \sum_i f_i s_i$, where the $s_i$ are in $\Gamma(M,B_1)$ and where the $f_i$ are smooth functions on $M$ vanishing at $x$. This implies that indeed $F(s - s')(x) = 0$. 

Secondly, $\phi$ is a $*$-algebra bundle homomorphism, since
\begin{eqnarray}
\phi(pq) &=& F(st)(x) = F(s)F(t)(x) = F(s)(x)\cdot F(t)(x) = \phi(p) \phi(q), \nonumber \\
\phi(p^*) &=& F(s^*)(x) = (F(s))^*(x) = (F(s)(x))^* = \phi(p)^*. \nonumber
\end{eqnarray}
where $s,t \in \Gamma(M,E)$ are such that $p=s(x)$, $q=t(x)$. 

Finally, by construction $(\hat{\Gamma}(\phi)s)(x) = \phi(s(x)) = F(s)(x)$, so that $\hat{\Gamma}(\phi) = F$ as required. Hence, $\hat{\Gamma}$ is a full functor.

\begin{rmk}
If $B_1$, $B_2$ are unital ($*$-)algebra bundles and $\phi: B_1 \rightarrow B_2$ is a unital ($*$-)algebra bundle homomorphism, then $\hat{\Gamma}(\phi)$ is a unital
($*$-preserving) $\cinf{M}$-algebra-homomorphism. Conversely, if $F: \Gamma(M,B_1) \rightarrow \Gamma(M,B_2)$ is a unital ($*$-preserving)
$\cinf{M}$-algebra-homomorphism, and $\phi: B_1 \rightarrow B_2$ is defined by (\ref{eq:full}), then
\begin{equation*}
\phi(1_x) = F(1)(x) = 1_{x}.
\end{equation*}
\end{rmk}

We summarize the results in this subsection in the following theorem.
\begin{thm}[Serre--Swan for algebra bundles]
\label{thm:serreswan2}
Let $M$ be a compact manifold. The functor $\hat{\Gamma}$ furnishes an equivalence between the category of (unital) (involutive) algebra bundles over $M$ and the category of (unital) finitely
generated projective (involutive) $\cinf{M}$-algebras. 
\end{thm}

\subsection{Spectral triple obtained from an algebra bundle}
\label{ssct:algebrabundletospectraltriple}
In this subsection we construct a real and even spectral triple whose algebra is isomorphic to $\Gamma(M,B)$. 
Here $B$ is some locally trivial $*$-algebra bundle whose fibers are copies of a fixed (finite-dimensional) $*$-algebra $A$. Furthermore, we require that for each $x$ the fiber $B_x$ is endowed with a faithful tracial state $\tau_x$ so that for all $s \in \Gamma(M,B)$ the function $x \mapsto \tau_x{s(x)}$ is smooth. The corresponding Hilbert--Schmidt inner product on the fiber $B_x$ induced by $\tau_x$ is denoted by $\langle \cdot , \cdot \rangle_{B_x}$. Consequently, the $\cinf{M}$-valued form
\begin{equation*}
( \cdot, \cdot )_B: \Gamma(M,B) \times \Gamma(M,B) \rightarrow \cinf{M}, \quad (s,t)_B(x) = \langle s(x), t(x) \rangle_{B_x},
\end{equation*}
turns $\Gamma(M,B)$ into a pre-Hilbert $\cinf{M}$-module. 

As in the previous sections, we assume that $M$ is a Riemannian spin manifold, on which $S \to M$ is a spinor bundle and $\Dslash = c \circ \nabla^S$ a Dirac operator. Combining the inner product on spinors with the above hermitian structure naturally induces the following inner product on $\Gamma(M,B \otimes S)$:
\begin{equation}
\langle \xi_1, \xi_2  \rangle_{\Gamma(M,B \otimes S)} := \int_M \langle \xi_1(x) , \xi_2(x) \rangle_{B_x \otimes S_x} \quad (\xi_1,\xi_2 \in \Gamma(M,B \otimes S)),
\label{eq:innerproductbundle}
\end{equation}
turning it into a pre-Hilbert space. The completion with respect to the norm induced by this inner product consists of all square-integrable sections of $B \otimes S$, and is denoted by $L^2(M,B \otimes S)$. 
\begin{rmk}
Note that we can identify $\Gamma(M,B \otimes S ) \cong \Gamma(M,B) \otimes_{\cinf{M}} \Gamma(M,S)$ as $\cinf{M}$-modules. In what follows, we will use this isomorphism without further notice. The above inner product (\ref{eq:innerproductbundle}) can be written as
\begin{equation*}
\langle  s_1 \otimes \psi_1, s_2 \otimes \psi_2 \rangle_{\Gamma(M,B) \otimes_{\cinf{M}} \Gamma(M,S)}  = \langle \psi_1 , (s_1,s_2)_B \psi_2 \rangle_S,
\end{equation*}
where $(s_1, s_2)_B \in \cinf{M}$ acts on $\Gamma(M,S)$ by point-wise multiplication. 
\end{rmk}

%Moreover, we assume that for all $x \in M$ the inner product
%on $B_x$ satisfies $\langle rp,q  \rangle_{x} = \langle p, r^*q \rangle_{x}$ for all $p,q,r \in B_x$.

%\begin{dfn}
%A $*$-algebra bundle $B$ is called an hermitian $*$-algebra bundle if for every $x \in M$ there is given an inner product $\langle \cdot, \cdot \rangle_{x}$ on $B_x$ such that for all sections $s,t \in \Gamma(M,B)$ the function $x \mapsto \langle s(x), t(x) \rangle_{x}$ is smooth, and for all $x \in M$ the inner product on $B_x$ satisfies $\langle rp,q  \rangle_{x} = \langle p, r^*q \rangle_{x}$ for all $p,q,r \in B_x$.
%\end{dfn}

\begin{thm}
\label{thm:spectraltriple1}
In the above notation, let $\nabla^B$ be a hermitian connection (with respect to the Hilbert--Schmidt inner product) on the algebra bundle $B$ and let $D_B = c \circ(\nabla^B \otimes 1 + 1 \otimes \nabla^S) $ be the twisted Dirac operator on $B \otimes S$. Then
$$
(\Gamma(M,B),L^2(M,B \otimes S), D_B) 
$$
is a spectral triple with summability equal to the dimension of $M$. 
\end{thm}
\begin{proof}
First, it is obvious that fiber-wise multiplication of $a \in \Gamma(M,B)$ on $\Gamma(M,B \otimes S)$ extends to a bounded operator on $L^2(M,B \otimes S)$ since
$$
\|as \otimes \psi \|^2 = \int_M \langle \psi(x), \langle a(x)s(x), a(x)s(x) \rangle_{B_x}  \psi(x) \rangle_{S_x} dx 
		\leq \sup_{x \in M} \{|a(x)|^2\} %\int_M \langle \psi(x),\langle s(x), s(x)\rangle_{B_x} \psi(x) \rangle_{S_x} dx = m^2 
\| s \otimes \psi \|^2.
$$
Compactness of the resolvent and summability is clear from ellipticity of the twisted Dirac operator $D_B$, $M$ being a compact manifold. Moreover, the commutator $[D_B,a]$ is bounded for $a \in \Gamma(M,B)$ since $D_B$ is a first-order differential operator. More precisely, in local coordinates one computes
\begin{equation*}
[D_B,a] (s \otimes \psi) = \left(\partial_{\mu} a + [\omega^B_{\mu},a]\right) s \otimes \gamma^{\mu} \psi.
%\label{eq:actionislocal}
\end{equation*}
where $\nabla^B_\mu = \partial_\mu + \omega^B_\mu$, locally. This operator is clearly bounded on $L^2(M,B \otimes S)$, provided $a$ is differentiable and $\omega^B_\mu$ is a smooth connection one-form.
\end{proof}

Next, we would like to extend our construction to arrive at a real spectral triple. For this, we introduce an anti-linear operator on $L^2(M,B \otimes S)$ of the form
$$
J ( s \otimes \psi ) = s^* \otimes J_M \psi
$$
with $J_M$ charge conjugation on $M$ (\textit{cf}. Example \ref{ex:canonical}). For this operator to be a real structure on our spectral triple $(\Gamma(B), L^2(B \otimes S), D_B)$, we need some extra conditions on the connection $\nabla_B$ on $B$. 
\begin{dfn}
Let $B$ be a $*$-algebra bundle over a manifold $M$. A $*$\emph{-algebra connection} $\nabla$ on $B$ is a connection on $B$ that satisfies
\begin{eqnarray*}
\nabla(st) = s\nabla t + (\nabla s)t, \qquad
(\nabla s)^* = \nabla s^*; \qquad (s,t \in \Gamma(M,B)).
\end{eqnarray*}
If $B$ is a hermitian $*$-algebra bundle and $\nabla$ is also a hermitian connection, then $\nabla$ is called a \emph{hermitian} $*$-algebra
connection.
\end{dfn}
Before we proceed we need to know whether a hermitian $*$-algebra connection exists on any given locally trivial $*$-algebra bundle. A partition of unity argument easily shows how to construct hermitian $*$-algebra connections on arbitrary $*$-algebra bundles.
\begin{lem}
Every locally trivial hermitian $*$-algebra bundle $B$ defined over a paracompact space $M$ admits a hermitian $*$-algebra connection.
\label{lma:existenceconnection}
\end{lem}
\begin{proof}
Let $\{U_i\}$ be a locally finite open covering of $M$ such that $B$ is trivialized over $U_i$ for each $i$. Then on each $U_i$ there exists a hermitian $*$-algebra connection $\nabla_i$, for instance the trivial connection $d$ on $U_i$. Now, let $\{f_i\}$ be a partition of unity subordinate to the open covering $\{U_i\}$ (all $f_i$ are real-valued). Then the linear map $\nabla$ defined by
\begin{equation*}
(\nabla s)(x) = \sum_i f_i(x) (\nabla_i s)(x), \quad (x \in M),
\end{equation*}
is a hermitian $*$-algebra connection on $\Gamma(M,B)$.
\end{proof}

\begin{rmk}
\label{rmk:derivation1form}
The fact that locally, on some trivializing neighborhood, the exterior derivative $d$ is a hermitian $*$-algebra connection shows that on such a local patch every hermitian $*$-algebra connection is of the form
\begin{equation*}
d + \omega,
\end{equation*}
where $\omega$ is a real connection 1-form with values in the real Lie algebra of $*$-derivations of the fiber that are anti-hermitian with respect to the inner product on the fiber. For instance, when the fiber is the $*$-algebra $M_N(\mathbb{C})$ endowed with the Hilbert--Schmidt inner product, this Lie algebra is precisely $\text{ad}(\mathfrak{u}(N)) \cong \mathfrak{su}(N)$.
\end{rmk}

\begin{thm}
\label{thm:spectraltriple2}
Suppose in addition to the conditions of Theorem \ref{thm:spectraltriple1} that $\nabla^B$ is hermitian $*$-algebra connection and set $\gamma_B = 1 \otimes \gamma_5$ as a self-adjoint operator on $L^2(M,B \otimes S)$. Then
$\spectraltripleJeven$ is a real and even spectral triple with $KO$-dimension equal to the dimension of $M$.
\end{thm}
\begin{proof}
First of all, we check that $J$ is anti-unitary:
%is well-defined with respect to the $\cinf{M}$-linearity because $J_M f = \bar{f} J_M$ for all $f \in \cinf{M}$. Since $T^2=1$ and $J_M^2=-1$ it follows that $J^2=-1$. Moreover $T \otimes J_M$ is anti-unitary, since for all $s,t \in \Gamma(M,B)$, $\psi, \eta \in \Gamma(M,S)$:
\begin{align*}
\langle J( s\otimes \psi), J (t\otimes \eta) \rangle &= \langle J_M \psi, (s^*,t^*) J_M \eta \rangle = \langle J_M \psi, J_M \ol{(s^*,t^*)}\eta
\rangle \nonumber \\
				&= \langle \ol{(s^*,t^*)} \eta, \psi \rangle = \langle (s, t) \eta, \psi \rangle = \langle t \otimes \eta, s \otimes \psi \rangle, \nonumber
\end{align*}
where we have in the second step that $J_M f = \bar{f} J_M$ for every $f \in \cinf{M}$, in the third step that $J_M$ is anti-unitary and in the fourth step that $(s,t) =
(t^*,s^*)$ (by definition of the hermitian structure as a fiber-wise trace). Moreover, since $J_M^2=-1$ it follows that $J^2=-1$. 

We next establish $DJ =JD$ by a local calculation:
\begin{eqnarray}
(JD - DJ)(s \otimes \psi) &=&J (\nabla^B_{\mu}s \otimes i\gamma^{\mu}\psi + s \otimes \Dslash \psi) - D_B(s^* \otimes J_M \psi)  \nonumber \\
			&=& (\nabla^B_{\mu} s)^* \otimes (-i)J_M \gamma^{\mu} \psi + s^* \otimes J_M \Dslash \psi - \nabla^B_{\mu} s^* \otimes i\gamma^{\mu} J_M \psi
			-s^* \otimes \Dslash J_M \psi \nonumber \\
			&=& -i\left( (\nabla^B_{\mu} s)^* - \nabla^B_{\mu} s^*  \right) \otimes J_M \gamma^{\mu} \psi = 0 \nonumber,
\end{eqnarray}
since in four dimensions $\{J_M, \gamma^{\mu}\}=[\Dslash, J_M]=0$, and the last step is established by the condition of a $*$-algebra connection, {\it i.e.} $(\nabla s)^* = \nabla s^*$ for all $s \in
\Gamma(M,B)$.

The commutant property follows easily:
\begin{eqnarray}
[a,b^0](s \otimes \psi) &=& aJb^*J^*(s \otimes \psi) - Jb^*J^*a(s \otimes \psi) = aJ(b^*s^* \otimes J_M^* \psi) - Jb^* (s^*a^* \otimes J_M^* \psi) \nonumber \\
&=& asb \otimes \psi - asb \otimes \psi = 0, \nonumber
\end{eqnarray}
where $a,b \in \Gamma(M,B)$ and $s \otimes \psi \in \Gamma(M,B) \otimes_{\cinf{M}} \Gamma(M,S)$. Since $[a,b^0]=0$ on $\Gamma(M,B) \otimes_{\cinf{M}} \Gamma(M,S)\cong
\Gamma(M,B \otimes S)$, it is zero on the entire Hilbert space $L^2(M,B\otimes S)$. It remains to check the order one condition for the Dirac operator. First note that
\begin{equation*}
[[D,a],b^0](s \otimes \psi) = i c([[\nabla, a], b^0]( s\otimes \psi)) \quad (a,b,s \in \Gamma(M,B)).
\end{equation*}
This is zero because $[[\nabla, a], b^0]( s\otimes \psi)$ is zero:
\begin{eqnarray}
([\nabla,a]sb) \otimes \psi - Jb^*J^* ([\nabla, a] s \otimes \psi)
&=& \nabla(asb) \otimes  \psi - a \nabla(sb) \otimes \psi - Jb^* J^* \nabla(as) \otimes \psi  + Jb^*J^* a (\nabla s) \otimes \psi \nonumber \\
&=& \nabla(asb) \otimes \psi - a \nabla(sb) \otimes \psi - \nabla(as)b \otimes \psi + a (\nabla s) b \otimes \psi \nonumber \\
&=& \left( (\nabla a) sb + a (\nabla s)b + as (\nabla b) - a (\nabla s)b \right.\nonumber \\
&& \quad - \left. as (\nabla b) - (\nabla a)sb - a(\nabla s)b + a (\nabla s)b \right) \otimes \psi, \nonumber \\
&=& 0 \nonumber
\end{eqnarray}
using the defining property for $\nabla^B$ to be a $*$-algebra connection.
Thus, $J$ fulfils all of the necessary conditions of a real structure on $(\Gamma(M,B), L^2(M,B \otimes S), D_B)$. The conditions on $\gamma_B$ to be a grading operator for this spectral triple are easily checked.
\end{proof}

In the next section we show that the triple $\spectraltripleJeven$ gives a non-trivial Yang--Mills theory over the manifold $M$. The Serre--Swan Theorem \ref{thm:serreswan2} plays an essential role in the proof. First, we explore the form of this spectral triple in the context of Kasparov's KK-theory.

\subsection{Relation with the unbounded Kasparov internal product}
\label{sect:KK}
\label{ssct:kkproduct}
In this section we establish that the spectral triple of Theorem \ref{thm:spectraltriple1} is an unbounded Kasparov product of two unbounded KK-cycles \cite{KasparovKK, BaajJulg}. Let us briefly recall some elementary notions from (unbounded) KK-theory. Denote by $\mc{B}(E)$ the bounded endomorphisms of a right Hilbert $B$-module $E$ and by $\mc{K}(E)$ the compact endomorphisms. 

\begin{dfn}
Let $A$ and $B$ be $\mathbb{Z}_2$-graded $C^*$-algebras. A Kasparov $A$-$B$-module consists of a triple $(E, \phi, F)$ where $E$ is a countably generated $\mathbb{Z}_2$-graded Hilbert-$B$-module, $\phi$ is a graded $*$-homomorphism $A \rightarrow \mc{B}(E)$ and $F$ is a bounded operator of degree $1$, such that $[F, \phi(a)]$, $(F^2 -1)\phi(a)$, and $(F - F^*)\phi(a)$ are in $\mc{K}(E)$.
\end{dfn}
There are the natural notions of unitary and homotopy equivalence and under the direct sum the set of equivalence classes of Kasparov $A-B$-modules forms an abelian group which is denoted by $KK(A,B)$ \cite{Kasparov}. 
%We say that two Kasparov $A$-$B$-modules $(E_1, \phi_1, F_1)$, $(E_2, \phi_2, F_2)$ are (unitarily) isomorphic if there exists a grading-preserving isomorphism of Hilbert-$B$-modules $\beta: E_1 \rightarrow E_2$ that intertwines the actions of $\phi_1$ and $\phi_2$ and those of $F_1$ and $F_2$.
%The set of all Kasparov $A$-$B$-modules is denoted by $\mathbb{E}(A,B)$. By $KK(A,B)$ we denote the set of equivalence classes of $\mathbb{E}(A,B)$ under the equivalence relations $\sim_h$, called homotopy, where $(E_0, \phi_0, F_0) \sim_h (E_1, \phi_1, F_1$ if there exists an element $(E,\phi,F) \in \mathbb{E}(A,IB)$ for which $E \otimes_{\pi_0} B \cong E_0$ and $E \otimes_{\pi_1} B \cong E_1$ as Kasparov $A-B$-modules. Here $\pi_i: IB \rightarrow B$ denotes evaluation at the point $i (=0,1)$. \textbf{Hier moet waarschijnlijk meer uitleg bij!} One defines an addition on $KK(A,B)$ by defining $[(E_1, \phi_1, F_1)] + [(E_2, \phi_2, F_2)] := [E_1 \oplus E_2, \phi_1 \oplus \phi_2, F_1 \oplus F_2]$.
%\begin{prp}[\cite{Kasparov}]
%The set $KK(A,B)$ with the addition as defined above forms an abelian group. 
%\end{prp}
One of the key properties of $KK$-theory is the existence of the internal Kasparov product.
\begin{dfn}
\label{dfn:boundedkasparovproduct}
Let $E_1$ be an $A$-$B$-module and $E_2$ a $B$-$C$-module, and define an $A$-$C$-module by $E:= E_1 \otimes_B E_2$. A Kasparov module $(E,\phi,F)$ is called a \emph{Kasparov product} for $(E_1, \phi_1, F_1)$ and $(E_2, \phi_2, F_2)$ if
\begin{itemize}
\item $(E, \phi_1 \otimes \text{Id}, F) \in KK(A,C)$;
\item for every $x \in E_1$ of homogeneous degree $\#x$, the operator $T_x: E_2 \rightarrow E$ defined by $T_x(e) = x\otimes e$ satisfies
\begin{eqnarray}
T_x \circ F_2 - (-1)^{\#x} F \circ T_x \in \mc{K}(E_2, E), \nonumber \\
F_2 \circ T_x^* - (-1)^{\#x} T^*_x \circ F \in \mc{K}(E,E_2);\nonumber
\end{eqnarray}
\item for all $a \in A$ the graded commutator $\phi(a)[F_1 \otimes \text{Id}, F]\phi(a^*) \geq 0 \text{ mod }\mc{K}(E)$.
\end{itemize}
\end{dfn}
If $A$ is separable and $B$ is $\sigma$-unital, then there is a Kasparov-product for $(E_1, \phi_1, F_1)$ and $(E_2, \phi_2, F_2)$ and any of these products are homotopic.\footnote{Actually, they are even operator homotopic (\textit{cf}. \cite{Kasparov} or \cite{Blackadar}).} Therefore, the internal Kasparov product defines a bilinear map $\otimes_B: KK(A,B) \times KK(B,C) \rightarrow KK(A,C)$.

The Kasparov internal product of Definition \ref{dfn:boundedkasparovproduct} can be captured in terms of unbounded Kasparov-modules \cite{BaajJulg}. %It is this product we are interested in since we will show that the spectral triple $\spectraltripleJeven$ of Theorem \ref{thm:spectraltriple} can be constructed as such an unbounded Kasparov internal product.

\begin{dfn}[\cite{BaajJulg}, \cite{KucerovskyUnboundedKKmodules}]
Let $A$ and $B$ be graded $C^*$-algebras. An \emph{unbounded Kasparov module} is a triple $(E, \phi, D)$ where $E$ is a graded Hilbert-$B$-module, $\phi: A \rightarrow \mc{B}(E)$ a graded $*$-homomorphism, and $D$ a self-adjoint regular operator in $E$, homogeneous of degree $1$ such that
$(1 + D^2)^{-1} \phi(a)$ extends to an element of $\mc{K}(E)$ for all $a \in A$, and the set of all $a \in A$ such that $[D,\phi(a)]$ extends to an element in $\mc{B}(E)$ is dense in $A$.
\end{dfn}
The set of all unbounded Kasparov modules is denoted by $\Psi(A,B)$. 
\begin{example}
The canonical spectral triple $(\cinf{M}, L^2(M,S), \Dslash, \gamma^5)$ is an element in $\Psi(\cinf{M},\mathbb{C})$. Another example is given as follows: let $B$ be a locally trivial $*$-algebra bundle with a smoothly-varying faithful tracial state on the fibers. Then $(\Gamma(M,B), L^2(M,B), 0)$ is an element of $\Psi(\Gamma(M,B),\cinf{M})$, and even in $KK(\Gamma(M,B),\cinf{M})$. Note that the algebras $\Gamma(M,B)$ and $\cinf{M}$ are trivially graded.
\end{example}

\begin{prp}[\cite{BaajJulg}]
If $(E,\phi,D) \in \Psi(A,B)$ then $(E, \phi, F) \in \mathbb{E}(A,B)$ where $F = D(1+D^2)^{-1}$. If $A$ is separable, the map $(E,\phi,D) \mapsto [(E, \phi, D(1+D^2)^{-1})]$ is a surjective map $\Psi(A,B) \rightarrow KK(A,B)$.
\end{prp}

Thus, classes in $KK(A,B)$ can be represented by unbounded cycles in $\Psi(A,B)$. The following theorem is due to Kucerovsky \cite{KucerovskyUnboundedKKmodules} and introduces a Kasparov product for unbounded $KK$-modules. This was further worked out by Mesland \cite{Mesland}.
\begin{thm}[Kucerovsky]
\label{thm:unboundedkasparovproduct}
Suppose that $(E_1 \otimes_B E_2, \phi_1 \otimes_B \text{Id}, D) \in \Psi(A,C)$, $(E_1, \phi_1, D_1) \in \Psi(A,B)$ and $(E_2, \phi_2, D_2) \in \Psi(B,C)$ are such that
\begin{enumerate}
\item for all $x$ in some dense subset of $\phi_1(A)E_1$, the operator
\begin{equation*}
\left[ \left( \begin{array}{cc}
D & 0 \\
0 & D_2
\end{array}
\right), \left(
\begin{array}{cc}
0 & T_x \\
T_x^* & 0
\end{array}
\right)\right]
\end{equation*}
is bounded on $\text{Dom }D \oplus \text{Dom }D_2$;
\item The resolvent of $D$ is compatible with $D_1$: that is, there is a dense submodule $W$ such that $D_1(i\mu + D)^{-1}(i\mu_1 + D_1)^{-1}$ is defined on $W$ for all $\mu, \mu_1 \in \mathbb{R} - \{0\}$;
\item There exists a $c \geq 0$ such that $\langle D_1 x, Dx \rangle + \langle Dx , D_1 x \rangle \geq c \langle x,x\rangle$, for all $x$ in the domain;
\end{enumerate}
where $x \in E_1$ is homogeneous and $T_x: E_2 \rightarrow E$ maps $e \mapsto x \otimes_B e$. Then $(E_1 \otimes_B E_2, \phi_1 \otimes_B \text{Id}, D)$ represents the Kasparov-product of $(E_1, \phi_1, D_1) \in \Psi(A,B)$ and $(E_2, \phi_2, D_2) \in \Psi(B,C)$.
\end{thm}

Using Theorem \ref{thm:unboundedkasparovproduct} we show that the spectral triple $\spectraltriple$ of Theorem \ref{thm:spectraltriple1} can be considered as a Kasparov product.

\begin{prp}
Let $B$ be a locally trivial hermitian unital $*$-algebra bundle on a compact Riemannian spin manifold $M$ with fibers isomorphic to some complex $*$-algebra $A$. Let $\nabla^B$ be a hermitian connection on $B$ and $D_B$ the corresponding twisted Dirac operator. Then 
$\spectraltriple$ is an unbounded Kasparov product of $(L^2(M,B), \lambda, 0) \in \Psi(\Gamma(M,B),\cinf{M})$ and $(L^2(M,S), m, \Dslash) \in \Psi(\cinf{M}, \mathbb{C})$, where $\lambda$ is the representation of $\Gamma(M,B)$ on $L^2(M,B)$ induced by left-multiplication and  where $m$ denotes the representation of $\cinf{M}$ on $L^2(M,S)$ by point-wise multiplication with elements in $\Gamma(M,S)$. 
\end{prp}
\begin{proof}
Most of the assertions are straightforward to prove. To prove the last statement we will check the first condition of Theorem \ref{thm:unboundedkasparovproduct} since the other two are trivial (because $D_1 = 0$). It suffices to check that
\begin{align*}
D \circ T_a - T_a \Dslash &\in \mc{B}(L^2(M,S), L^2(M, B \otimes S)), \\ 
\Dslash T_a^* -  T_a^* D &\in \mc{B}(L^2(M,B\otimes S), L^2(M,S)) ,
\end{align*}
for all $a \in \Gamma(M,B)$. For the first condition, we have for $\psi \in L^2(M,S)$ that
\begin{equation*}
(D \circ T_a - T_a D_2)(\psi) = D(a \otimes \psi) - a \otimes \Dslash \psi = c(\nabla^B a) \otimes \psi 
%\leq \|\langle \gamma(\nabla a), \gamma(\nabla ^Ba) \rangle \| \|\psi\|,
\end{equation*}
so that $D \circ T_a - T_a \Dslash$ extends to a bounded operator. Now the second one:
\begin{align*}
(\Dslash T_a^* -  T_a^* D) (s \otimes \psi) &= \Dslash ( \langle a, s \rangle \psi ) - \langle a, s \rangle \Dslash \psi - \langle a, c(\nabla^B s) \rangle \psi \\
&= [\Dslash, \langle a ,s \rangle]\psi - ( [\Dslash, \langle a,s\rangle] - \langle c(\nabla^B a),  s \rangle)\psi \\
&= \langle c(\nabla^B a),  s \rangle \psi, 
\end{align*}
which is again uniformly bounded. This completes the proof.
\end{proof}
Another proof of this fact follows by adopting the direct construction of the unbounded Kasparov products by Mesland \cite{Mesland}. Indeed, the spectral triple $\spectraltriple \in \Psi(\cinf{M}, \C)$ is by construction the internal product of $(\Gamma(M,B), L^2(M,B), 0) \in \Psi(\Gamma(M,B),\cinf{M})$ and $(\cinf{M}, L^2(M,S), \Dslash, \gamma^5) \in \Psi(\cinf{M},\C)$.

\section{Yang--Mills theory as a noncommutative manifold}
\label{sect:ym}
The spectral triple $\spectraltripleJeven$ that we obtained in Theorem \ref{thm:spectraltriple2} will turn out to be the correct triple to describe a non-trivial $PSU(N)$-gauge theory on the manifold $M$ if the fibers of $B$ are taken to be isomorphic to the $*$-algebra $M_N(\mathbb{C})$. Not only does it describe a non-trivial $PSU(N)$-gauge theory, every $PSU(N)$-gauge theory on $M$ is described by such a triple. In this section we will prove these claims by first showing how a principal $PSU(N)$-bundle can be constructed from this spectral triple (in fact, the algebra $\Gamma(M,B)$ is already sufficient for this). As in the topologically trivial case \cite{ChamseddineConnes}, the spectral action applied to this triple will give the Einstein--Yang--Mills action, but now the gauge potential can be interpreted as a connection 1-form on the $PSU(N)$-bundle $P$. In fact, the original algebra bundle $B$ will turn out to be an associated bundle of the principal bundle $P$. From now on, the fibers of $B$ are assumed to be $M_N(\mathbb{C})$.

%The idea behind the proof is the following. Since the representation of $PSU(N)$ on $M_N(\mathbb{C})$ is faithful, one can construct the bundle $P$ from the $*$-algebra bundle $B$ (and viceversa)\footnote{This is the reason why a $PSU(N)$ is constructed and not a $SU(N)$-bundle. The adjoint action of $SU(N)$ is not a faithful representation because $SU(N)$ has nontrivial centre.} It is therefore sufficient to show that the $*$-algebra bundle $B$ can be obtained from its space of smooth sections $\Gamma(M,B)$ which we already proved in Theorem \ref{thm:serreswan2}. We now make this more precise.

\subsection{From the spectral triple to principal bundles}
According to Theorem \ref{thm:serreswan2} we are able to reconstruct the unital $*$-algebra bundle $B$ from $\Gamma(M,B)$. %All the information of $B$ is therefore also contained in $\Gamma(M,B)$, regarded as a involutive finitely generated unital $\cinf{M}$-algebra. 
Note that in this theorem the ($*$-)algebra bundles are not required to be locally trivial as a ($*$-)algebra bundle (they \emph{are} locally trivial as a \emph{vector} bundle). For the rest of this section we assume that $B$ is a locally trivial $*$-algebra bundle with fiber $M_N(\mathbb{C})$.

In order to construct a principal $PSU(N)$-bundle $P$ out of $B$, first of all note that since all $*$-automorphisms of $M_N(\mathbb{C})$ are obtained by conjugation with a unitary element $u \in M_N(\mathbb{C})$ the transition functions of the bundle $\Gamma(M,B)$ have their values in $\text{Ad }U(N) \cong U(N) / Z(U(N)) \cong PSU(N)$. Thus the bundle $B$ provides us with a open covering of $\{U_i\}$ and transition functions $\{g_{ij}\}$ with values in $PSU(N)$. Using the reconstruction theorem for principal bundles we can construct a principal $PSU(N)$-bundle. By construction, the bundle $B$ is an associated bundle to $P$.

Furthermore, for the real and even spectral triple $\spectraltripleJeven$ of Theorem \ref{thm:spectraltriple2} the hermitian connection $\nabla^B$ on the bundle $B$ can locally be written as $\nabla^B = d + \omega$, where $\omega$ is a $\mathfrak{su}(N)$-valued 1-form, (\textit{cf}. Eq. \ref{rmk:derivation1form}). Moreover, the transformation rule for $\omega$ is $\omega_i = g^{-1}_{ij} dg_{ij}  + g^{-1}_{ij} \omega_j g_{ij}$ with $g_{ij}$ the $PSU(N)$-valued transition function of $B$. Comparing this expression with the transformation property of a connection 1-form in Definition \ref{dfn:connection1form} one concludes that the hermitian $*$-algebra connection $\nabla^B$ on $B$ induces a connection 1-form on the principal bundle $P$ constructed in the previous paragraph (and vice versa).

Conversely, given a $PSU(N)$-gauge theory $(P, \omega)$ on some compact Riemannian spin-manifold, then we can construct the locally trivial hermitian $*$-algebra bundle $B:=P \otimes_{PSU(N)} M_N(\mathbb{C})$, where $PSU(N)$ acts on $M_N(\mathbb{C})$ in the usual way. Moreover, the connection $\omega$ on $P$ induces a hermitian $*$-algebra connection on $B$. By following the steps in the previous section it is not difficult to see that the gauge theory $(P_B, \omega_B)$ obtained from this spectral triple $(\Gamma(M,B), L^2(M,B\otimes S), c(\nabla^B \otimes 1 + 1 \otimes \nabla^S), J, \gamma)$ is isomorphic to $(P,\omega)$. 
This is in accordance with the approach to almost noncommutative manifolds taken in \cite{Branimir}.

\begin{prp}
\label{cor:reconstructiontheorempfb}
Let $\spectraltripleJeven$ be as before, with $B$ an endomorphism bundle. Then
% where $B$ is a locally trivial hermitian $*$-algebra bundle with fiber $M_N(\mathbb{C})$ and $\nabla^B$ a hermitian $*$-algebra connection thereon, 
there exists a principal $PSU(N)$-bundle $P$ such that $B$ is an associated bundle of $P$, and a connection 1-form $\omega$ on $P$. Moreover, every $PSU(N)$-gauge theory on $M$ is determined by such a spectral triple.
\end{prp}

%In the next section we will calculate the spectral action and explore what happens when the Dirac operator $D_B$ is fluctuated. We would like to obtain a Lagrangian where the connection 1-form $\omega$ occurs as the gauge potential in the Lagrangian.

\subsection{Spectral action}
In this section, we will calculate the spectral action for the real spectral triple of Theorem \ref{thm:spectraltriple2}. We will show that the spectral action applied to the spectral triple $\spectraltripleJeven$ produces the Einstein--Yang--Mills action for a 1-form $A$ that defines a connection 1-form on the $PSU(N)$-bundle $P$. If $B$ is a trivial algebra bundle, this reduces to the result of \cite{ChamseddineConnes}. In fact, much of their local computations can be adopted in this case as well, since locally, the bundle $B$ is trivial. Nevertheless, for completeness we include the computation in the case at hand.

First of all, already in Remark \ref{rmk:derivation1form} we noticed that locally, on some local trivialization $U_i$, $\nabla^B$ is expressed as $d + \omega_i$ where
$\omega$ is an $su(N)$-valued 1-form that acts in the adjoint representation on $\Gamma(M,B)$. Therefore, according to Definition \ref{dfn:connection1form} $\omega$ already
induces a connection 1-form on $P$. To get the full gauge potential we need to take the fluctuation of the Dirac operator into account as well.

Inner fluctuations of the Dirac operator are given by a perturbation term of the form
\begin{equation*}
A = \sum_j a_j[D,b_j], \qquad (a_j, b_j \in \Gamma(M,B)),
\label{eq:nontrivinnerfluc}
\end{equation*}
with the additional condition that $\sum_j a_j[D,b_j]$ is a self-adjoint operator. Explicitly, we have
\begin{equation*}
A = \sum_j c \circ (a_j [\nabla,b_j] \otimes 1),
\end{equation*}
where $c: \Omega^1(M) \otimes_{\cinf{M}} \Gamma(M, B \otimes S) \rightarrow \Gamma(M,B\otimes S)$ is given by
\begin{equation*}
c(\omega \otimes s \otimes  \psi) = s \otimes c(\omega)\psi, \quad (\omega \in \Omega^1(M), s \otimes \psi \in \Gamma(M,B\otimes S)).
\end{equation*}
Also, $\sum_j a_j [\nabla, b_j]$ is an element of $\Gamma(T^*M \otimes B)$.

Locally, on some trivializing neighborhood $U$, the expression in Eq. \eqref{eq:nontrivinnerfluc} can be written as
\begin{equation*}
A= \gamma^{\mu} A_{\mu},
\end{equation*}
where $A_{\mu}$ are the components of the 1-form $\sum_j a_j [\nabla, b_j]$ with values in $\Gamma(M,B)$. Since $A$ is self-adjoint the 1-form $A_{\mu}$ can be
considered as a real 1-form taking values in the hermitian elements $\Gamma(M,B)$.

Similarly, the expression $A + JAJ^*$ is locally written as
\begin{equation*}
\gamma^{\mu} A_{\mu} - \gamma^{\mu} J A_{\mu} J^*,
\end{equation*}
since $\gamma^{\mu}$ anti-commutes with $J$ in 4 dimensions. Writing out the second term gives:
\begin{equation*}
(\gamma^{\mu} J A_{\mu} J^*)(s\otimes \psi) = s A_{\mu} \otimes \gamma^{\mu} \psi,  \quad \forall s \otimes \psi \in \Gamma(M,B\otimes S).
\end{equation*}
so that on this local patch $A + JAJ^*$ can be written as
\begin{equation*}
\gamma^{\mu} \ad A_{\mu}.
\end{equation*}
Consequently, $A + JAJ^*$ eliminates the $i u(1)$-part of $A$, making it natural to impose the uni-modularity condition
\begin{equation*}
\text{Tr } A = 0.
\end{equation*}
Thus, $-i \ad A_{\mu}$ is a one-form on $M$ with values in $\Gamma(M,\ad P)$. We denote this 1-form by $- i \bA^{pert}$; it is defined on
the whole of $M$.

The local and global expression for $D + A + JAJ^*$ are given respectively by
\begin{equation*}
D_A = i\gamma^{\mu} (\nabla^B_{\mu} \otimes 1 + 1 \otimes \nabla^S_{\mu} - i \ad A_{\mu} \otimes 1)  \\
%\label{eq:localformda}
\end{equation*}
and
\begin{equation*}
D_A = ic \circ (1 \otimes \nabla^S + \nabla_B \otimes 1 + \bA^{pert}),
\end{equation*}

On some trivializing neighborhood $U_i$ ($i \in I$) the connection $\nabla^B$ can be expressed as $d + \bA_i^0$ for a unique $su(N)$-valued 1-form $\bA^0_i$
on $U_i$. Thus, on $U_i$ the fluctuated Dirac operator can be rewritten as
\begin{equation*}
%\label{eq:fluct-Dirac}
D_A = i c \circ (d + 1 \otimes \omega^s + (\bA_i^0 + \bA^{pert}) \otimes 1).
\end{equation*}
We interpret $(\bA_i^0 + \bA^{pert})$ as the full gauge potential on $U_i$; it acts in the adjoint representation on the spinors. The natural action of $g \in \textup{Inn}(\Gamma(M,B) \simeq \Gamma(M,\Ad P)$ by conjugation on $D_A$ induces the familiar gauge transformation:
\begin{equation*}
\mathbb{A}^0 + \mathbb{A}^{pert} \mapsto (g^{-1}\mathbb{A}^0g + g^{-1}(dg)) + g^{-1}\mathbb{A}^{pert}g = g^{-1} (dg) + g ( \mathbb{A}^0 + \mathbb{A}^{pert})g^{-1},
\end{equation*}
where the first two terms are the transformation of $A^0$ under a change of local trivialization, and the last term is the transformation of $A^{pert}$. Since $P$ is an
associated bundle of $B$ it follows from Definition $\ref{dfn:connection1form}$ that $\mathbb{A}^0 + \mathbb{A}^{pert}$ induces a $su(N)$-valued connection 1-form on the
principal $PSU(N)$-bundle $P$ that acts on $\Gamma(M,B)$ in the adjoint representation. Let us summarize what we have obtained so far.

\begin{prp}
Let $\spectraltripleJeven$ and $P$ be as before, so that $P \times_{PSU(N)} M_N(\C) \simeq B$. Then
\begin{enumerate}
\item The group of inner automorphisms $\textup{Inn}(\Gamma(M,B)) \simeq \Gamma(M,\Ad P)$ where $\Ad P = P \times_{PSU(N)} SU(N)$.
\item The inner fluctuations of $D_B$ are parametrized by a section $\bA^{pert}$ of $\Gamma(M,\ad P)$ where $\ad P = P \times_{PSU(N)} su(N)$.
\end{enumerate}
Moreover, the action of $\textup{Inn}(\Gamma(M,B))$ on the inner fluctations $D_B + A + JAJ^{-1}$ by conjugation coincides with the adjoint action of $\Gamma(M,\Ad P)$ on $\Gamma(M,\ad P)$.
\end{prp}

Let us now proceed to compute the spectral action for these inner fluctuations. First, we recall some results on heat kernel expansions and Seeley--DeWitt coefficients, which will be useful later on; for more details we refer to \cite{Gilkey}. 

If $V$ is a vector bundle on a compact Riemannian manifold $(M, g)$ and if $Q:C^{\infty}(V)\to C^{\infty}(V)$ is a second-order elliptic differential operator of the form
\begin{equation*}
 Q = - \big(g^{\mu\nu}\partial_{\mu}\partial_{\nu} + K^{\mu}\partial_{\mu} + L)\label{eq:elliptic} 
\end{equation*}
with $K^{\mu}, L \in \Gamma(\textup{End}(V))$, then there exist a unique connection $\nabla$ and an endomorphism $E$ on $V$ such that
$Q = \nabla\nabla^* - E$. 
In this situation we can make an asymptotic expansion (as $t \to 0$) of the trace of the operator $e^{-tQ}$ in powers of $t$:
\begin{equation*}
\tr\,e^{-tQ} \sim \sum_{n \geq 0}t^{(n-m)/2}a_n(Q),\qquad a_n(Q) := \int_{M}a_n(x, Q)\sqrt{g}d^m x\label{eq:gilkey},
\end{equation*}
where $m$ is the dimension of $M$ and the coefficients $a_n(x, Q)$ are called the \emph{Seeley--DeWitt coefficients}. It turns out \cite[Theorem~4.8.16]{Gilkey} that $a_n(x, Q) = 0$ for $n$ odd and that the first three even coefficients are given (modulo boundary terms) by
\begin{align*}
a_0(x,Q) &= (4 \pi)^{-m/2} \text{Tr}(\text{Id})  \\
a_2(x,Q) &= (4 \pi)^{-m/2} \text{Tr}(-\frac{R}{6}\text{Id} + E)  \\
a_4(x,Q) &= (4 \pi)^{-m/2} \frac{1}{360}\text{Tr}\left(
%-12R\indices{_{;\mu}^{;\mu}} + 
5R^2 - 2R^{\mu \nu} R_{\mu \nu} + 2R_{\mu \nu \rho \sigma} R^{\mu \nu \rho \sigma} - 60RE %\right.
+ 180E^2 + %\left.60F\indices{_{;\mu}^{;\mu}} + 
30\Omega_{\mu \nu} \Omega^{\mu \nu}
\right),
\end{align*}
where $\Omega_{\mu \nu}$ is the curvature of the connection $\nabla$. 

\medskip

This can be used in the computation of the spectral action as follows. Assume that the inner fluctuations give rise to an operator $D_A$ for which $D_A^2$ is of the form \eqref{eq:elliptic} on some vector bundle $V$ on a compact Riemannian manifold $M$. Then, on writing $f$ as a Laplace transform, we obtain
\begin{equation}
\label{eq:laplace}
f (D_A/\Lambda) = \int_{t>0} \tilde g(t) e^{-t D_A^2/\Lambda^2} ~ d t.
\end{equation}
One calculates that in 4 dimensions the heat expansion (up
to order $n=4$) of the spectral action (\ref{eq:spectralaction}) is given by
\begin{align*}
\text{Tr}\left(f(D/\Lambda)\right) &\sim f(0) \Lambda^0 a_4(D^2) + \sum_{n=0,2} \Lambda^{4-n} a_n(D^2) \frac{1}{\Gamma(\frac{4-n}{2})} \int_0^{\inf} k(v) v^{\frac{4-n}{2}-1}
dv \\
&= f(0) \Lambda^0 a_4(D^2) + 2 f_2 \Lambda^2 a_2(D^2) + 2\Lambda^4 f_4 a_0(D^2).
\label{eq:Diracexpansion}
\end{align*}
where the $f_k$ are moments of the function $f$: 
\begin{align*}
	f_{k} := \int_{0}^{\infty} f(w)w^{k-1}dw; \qquad (k>0) \nonumber. 
\end{align*}

\begin{lem}
\label{lem:ED2}
For the spectral triple $\spectraltripleJeven$, the square of the fluctuated Dirac operator $D_A^2$ is locally of the form $-g_{\mu \nu} \partial_{\mu} \partial_{\nu} + K_{\mu} \partial_{\mu} + L$
and we have the following expressions for $\Omega_{\mu\nu}$ and $E$:
\begin{eqnarray}
E = -\frac{1}{4} R \otimes 1_{N^2} - \sum_{\mu < \nu} \gamma^{\mu} \gamma^{\nu} \otimes F_{\mu \nu} \nonumber \\
\Omega_{\mu \nu} = \frac{1}{4} R^{ab}_{\mu \nu} \gamma_{ab} \otimes 1_{n^2} + id_4 \otimes F_{\mu \nu} \nonumber,
\end{eqnarray}
where $F_{\mu \nu}$ is the curvature of the connection $\nabla^B + \bA^{pert}$.
\end{lem}

This result allows us to compute the bosonic spectral action for the fluctuated Dirac operator $D_A$, essentially reducing the computation in terms of a local trivialization to the trivial case of \cite{ChamseddineConnes} with the following result.

%The difference with the trivial case is that the full gauge potential of the theory is not given by only the perturbation term of the Dirac operator. It also has a contribution from the connection $\nabla^B$. This is a result of the non-triviality of the bundle $B$. Since the calculation of the spectral action in \textbf{maak referenties} was done locally, we immediately get the following result.
\begin{thm}
\label{thm:actionnontriv}
For the spectral triple $\spectraltripleJeven$, the spectral action equals the Yang--Mills action for $\nabla^B + \bA^{pert}$ minimally coupled to gravity:
\begin{equation*}
\tr \left( f(D_A/\Lambda) \right) \sim \frac{f(0)}{24 \pi^2} \int_M \tr F_{\mu\nu} F^{\mu\nu} \sqrt{g} d^4x + \frac{1}{(4\pi)^2} \int_M \mathcal{L}(g_{\mu \nu}) \sqrt{g} d^4x + \mathcal{O}(\Lambda^{-2}),
\end{equation*}
where $\mathcal{L}(g^{\mu \nu})$ is given by
\begin{equation*}
\mathcal{L}(g^{\mu \nu}) = 2N^2 \Lambda^4 f_4 + \frac{N^2}{6} \Lambda^2 f_2 R  - \frac{N^2 f_0}{80} C_{\mu \nu \rho \sigma} C^{\mu \nu \rho
\sigma},
\end{equation*}
ignoring topological and boundary terms. Here $C$ denotes the Weyl-tensor and $f_i$ are the $i$'th moments of the function $f$. 
\end{thm}

\section{Conclusions and outlook}
We have generalized the noncommutative description of the Einstein--Yang--Mills system by Chamseddine and Connes \cite{CC97} to the case where the principal bundle describing the gauge field is non-trivial. We have obtained a spectral triple from an algebra bundle and related its construction to the internal Kasparov product in unbounded KK-theory. If the typical fiber of the algebra bundle is $M_N(\C)$, we have showed that its internal fluctuations are parametrized by a $PSU(N)$-gauge field. In fact, we reconstructed a $PSU(N)$-principal bundle for which the algebra bundle is an associated bundle, and on which the gauge field defines a connection one-form. Finally, we have applied the spectral action principle to these inner fluctuations of the spectral triple and derived the Yang--Mills action for a $PSU(N)$-gauge field, minimally coupled to gravity. 

A natural question that arises in this topologically non-trivial context is how to incorporate, besides the Yang--Mills action, a topological action functional. Given an (even) spectral triple $(\A,\H,D,\gamma)$, we introduce -- besides the spectral action \eqref{eq:spectralaction} -- an invariant by 
\begin{equation}
\label{eq:topspectralaction}
S_{\top} [A] = \tr \left( \gamma f(D_A/\Lambda) \right).
\end{equation}
We will call this the {\it topological spectral action}.
It is clearly invariant under the action of the group of unitaries in the algebra $\A$, acting on $\gamma$ by conjugation. 

If we again write $f$ as a Laplace transform \eqref{eq:laplace} and use the McKean--Singer formula,
$$
\tr e^{-t D_A^2} = \ind D_A,
$$
then we can prove that asymptotically
$$
S_{\top} [A] \sim f(0) \ind D_A
$$
In our case of interest, {\it i.e.} the setting of Theorem \ref{thm:actionnontriv}, we thus find with the Atiyah--Singer index theorem an extra contribution of the form
$$
S_{\top} [A] \sim \frac{f(0)}{(2 \pi i)^{n/2}} \int_M \hat{A}(M) \ch (B).
$$
in terms of the $\hat A$ genus of $M$ and the Chern character of the algebra bundle $B$.

%\bibliographystyle{plainmath}
%	\bibliography{myrefs,references}

\begin{thebibliography}{99}

\bibitem{Ati79}
M.~F. Atiyah.
\newblock {\em The Geometry of {Y}ang-{M}ills Fields}.
\newblock Fermi Lectures. Scuola Normale, Pisa, 1979.

\bibitem{BaajJulg}
S.~Baaj and P.~Julg.
\newblock Th\'eorie bivariante de Kasparov et op\'erateurs non born\'es dans
  les $C^*$-modules hilbertiens.
\newblock {\em Acad. Sci. Paris} 296 (1983)  875--878.

\bibitem{Blackadar}
B.~Blackadar.
\newblock {\em K-theory for Operator Algebras}.
\newblock Cambridge University Press, 1998.

\bibitem{Bleecker}
D.~Bleecker.
\newblock {\em Gauge theories and variational principles}.
\newblock Addison-Wesley Publishing Company, 1981.

\bibitem{Branimir}
B.~\'Ca\'ci\'c.
\newblock Almost-commutative spectral triples and noncommutative-geometric
  field theory.
\newblock Poster at Bonn International Graduate School in Mathematics, 2009.

\bibitem{CC97}
A.~H. Chamseddine and A.~Connes.
\newblock The spectral action principle.
\newblock {\em Commun. Math. Phys.} 186 (1997)  731--750.

\bibitem{ChamseddineConnes}
A.~Chamseddine and A.~Connes.
\newblock The spectral action principle.
\newblock {\em Comm. Math. Phys.} 186 (1997)  731--750.

\bibitem{CCM07}
A.~H. Chamseddine, A.~Connes, and M.~Marcolli.
\newblock {Gravity and the Standard Model with neutrino mixing}.
\newblock {\em Adv. Theor. Math. Phys.} 11 (2007)  991--1089.

\bibitem{Connes}
A.~Connes.
\newblock {\em Noncommutative geometry}.
\newblock Academic Press, London and San Diego, 1994.

\bibitem{ConnesGravity}
A.~Connes.
\newblock Gravity coupled with matter and the foundation of noncommutative
  geometry.
\newblock {\em Comm. Math. Phys.} 182 (1996).

\bibitem{ConnesMarcolli}
A.~Connes and M.~Marcolli.
\newblock {\em Noncommutative Geometry, Quantum Fields and Motives}.
\newblock AMS, Providence, 2008.

\bibitem{Gilkey}
P.~Gilkey.
\newblock {\em Invariance theory, the heat equation and the Atiyah-Singer index
  theorem}.
\newblock CRC Press, 1984.

\bibitem{Varilly}
J.~Gracia-Bondia, J.~V\'arilly, and H.~Figueroa.
\newblock {\em Elements of noncommutative geometry}.
\newblock Birkhauser Boston, 2001.

\bibitem{KasparovKK}
G.~Kasparov.
\newblock The operator K-functor and extensions of $C^*$-algebras.
\newblock {\em Izv. Akad. Nauk. SSSR} 44 (1980)  571--636.

\bibitem{Kasparov}
G.~Kasparov.
\newblock Equivariant KK-theory and the Novikov conjecture.
\newblock {\em Invent. Math.} 91 (1988)  147--201.

\bibitem{KucerovskyUnboundedKKmodules}
D.~Kucerovsky.
\newblock The KK-product of unbounded modules.
\newblock {\em K-theory} 11 (1997)  17--34.

\bibitem{Mesland}
B.~Mesland.
\newblock {\em Unbounded Bivariant K-Theory and Correspondences in
  Noncommutative Geometry}.
\newblock PhD thesis, Universit\"at Bonn, 2009.

\bibitem{Nestruev}
J.~Nestruev.
\newblock {\em Smooth manifolds and observables}.
\newblock Springer, 2003.

\bibitem{Rieffel}
M.~Rieffel.
\newblock Morita equivalence for ${C}^*$-algebras and ${W}^*$-algebras.
\newblock {\em J. Pure Appl. Algebra} 5 (1974)  51--96.

\bibitem{Swan}
R.~Swan.
\newblock Vector bundles and projective modules.
\newblock {\em Transactions of the American Mathematical Society} 105 (1962)
  264--277.

\end{thebibliography}
\newcommand{\noopsort}[1]{}\def\cprime{$'$}

\end{document}